\RequirePackage{xcolor}
\documentclass[journal,twoside,web]{ieeecolor}

\usepackage{amsthm}
\usepackage{defs}
\usepackage{cite}
\usepackage{etoolbox}
\makeatletter
\@ifundefined{color@begingroup}%
  {\let\color@begingroup\relax
   \let\color@endgroup\relax}{}%
\def\fix@ieeecolor@hbox#1{%
  \hbox{\color@begingroup#1\color@endgroup}}
\patchcmd\@makecaption{\hbox}{\fix@ieeecolor@hbox}{}{\FAILED}
\patchcmd\@makecaption{\hbox}{\fix@ieeecolor@hbox}{}{\FAILED}

\def\BibTeX{{\rm B\kern-.05em{\sc i\kern-.025em b}\kern-.08em
    T\kern-.1667em\lower.7ex\hbox{E}\kern-.125emX}}
\begin{document}
\pagestyle{empty}
\title{Stability of multiplexed NCS based on an epsilon-greedy algorithm for communication selection}
\author{Harsh Oza$^{1}$,  Irinel-Constantin Mor\u{a}rescu$^{2}$, Vineeth S. Varma$^{2}$, and Ravi Banavar$^{1}$
\thanks{$^{1}$ The authors are with Systems $\&$ Control Engineering, IIT Bombay,
Powai, Mumbai 400076, India.
Email:{\tt\small \{harsh.oza,banavar\} @iitb.ac.in}.
$^{2}$ The authors are with Universite de Lorraine, CNRS, CRAN,
F-54000 Nancy, France and also associated with the Automation Department, Technical University of Cluj-Napoca, 400114 Cluj-Napoca, Romania.
Email:  {\tt\small  \{constantin.morarescu, vineeth.satheeskumar-varma\}@univ-lorraine.fr}}
\thanks{The authors thank the support of the Indo-French Centre for the Promotion of Advanced Research (IFCPRA). The work of I.C. Mor\u{a}rescu$^{2}$, V.S. Varma has also been supported by project DECIDE funded under the PNRR I8 scheme by the Romanian Ministry of Research.}
\thanks{\textcolor{blue}{A preliminary version of this article has been submitted to IEEE Control Systems articles.}}
}
\maketitle
\thispagestyle{empty}
\begin{abstract}
In this letter, we study a Networked Control System (NCS) with multiplexed communication and Bernoulli packet drops. Multiplexed communication refers to the constraint that transmission of a control signal and an observation signal cannot occur simultaneously due to the limited bandwidth. First, we propose an $\varepsilon$-greedy algorithm for the selection of the communication sequence that also ensures Mean Square Stability (MSS). We formulate the system as a Markovian Jump Linear System (MJLS) and provide the necessary conditions for MSS in terms of Linear Matrix Inequalities (LMIs) that need to be satisfied for three corner cases. We prove that the system is MSS for any convex combination of these three corner cases. Furthermore, we propose to use the $\varepsilon$-greedy algorithm with the $\varepsilon$ that satisfies MSS conditions for training a Deep $\qfactor$ Network (DQN). The DQN is used to obtain an optimal communication sequence that minimizes a quadratic cost. We validate our approach with a numerical example that shows the efficacy of our method in comparison to the round-robin and a random scheme.
\end{abstract}



%
\section{Introduction}\label{sec: MSS intro}
 A Networked Control System (NCS) is a system consisting of a plant, a controller, and a communication network. NCS plays an important role in various industries such as chemical plants, power systems, and warehouse management \cite{bemporad2010networked,ge2017distributed}. In an NCS, the communication between the plant and the controller, as well as between the controller and the actuator, often occurs over a wireless network. Various challenges exist in such systems, namely, limited resources (bandwidth), delays, packet drops, or adversarial attacks \cite{elia2004bode,heemels2010networked,dolk2017event,sandberg2022secure}. These uncertainties can affect the plant's performance, making it important to study these scenarios. Our focus is on an NCS with limited bandwidth, formulated as a multiplexing constraint on transmitting control and observation signals and considering random packet drops. We aim to find a switching policy for selecting a communication direction that ensures Mean Square Stability (MSS) and then find the optimal communication sequence that minimizes a quadratic performance measure.

In an early piece of work, Athans et al. proposed an uncertainty threshold principle \cite{athans1977uncertainty}. This principle states that optimum long-range control of a dynamical system with uncertainty parameters is feasible if and only if the uncertainty does not exceed a certain threshold. This forms the basis of subsequent work which gives stability conditions in different scenarios. 
In \cite{imer2004optimal}, linear systems controlled over a network face packet drop uncertainties, and are modeled as Bernoulli random variables, with independent drops in both communication channels. The necessary conditions for MSS and optimal control solutions using dynamic programming are then provided. Schenato et al. generalize this work by considering noisy measurements and giving stronger conditions for the existence of the solution \cite{schenato2007foundations}. They also show that the separation principle holds for the TCP protocol. Other studies focus on the design of a Kalman filter for wireless networks \cite{liu2004kalman, sinopoli2004kalman} and stability conditions for systems with packet drops \cite{you2011mean}. An alternate approach for random packet drops involves sending multiple copies \cite{mesquita2012redundant}, and event-triggered policies for nonlinear systems with packet drops are developed in \cite{varma2023transmission}.

While stability and optimality problems are addressed in the literature, the stability and optimal scheduling problems with multiplexing in control and observation have not been completely investigated. For instance, \cite{schenato2007foundations} discusses multiplexing and packet drops but not optimal network selection, while \cite{maity2021optimal} considers a joint strategy for optimal selection and control of an NCS, but only for sensor signals. Major directions of research incorporate bandwidth constraints as i) multiplexing in multiple sensor signals, e.g., \cite{sinopoli2004kalman}, ii) multiplexing in sensor and control signals, e.g., \cite{schenato2007foundations}. In addition, other communication uncertainties such as packet drops and delays are addressed \cite{zhang2012network}. The main objective of these attempts is to find an optimal control strategy or to find an optimal policy for the selection of communication channels, or both see, e.g., \cite{molin2009lqg}. Leong et al. address the boundedness of error covariance objective in a multiplexed sensor information system with packet drops \cite{leong2020deep}. A stability condition is established based on the packet drop probability, and then an optimal sequence of communication is found by training a DQN with a $\varepsilon$
- greedy algorithm.

In this letter, our contributions are as follows:
\begin{enumerate}[label = \roman*)]
    \item We propose a modified $\varepsilon$-greedy algorithm for the selection of the direction of communication (transmit or receive.)
    \item We establish the necessary conditions for the MSS of an NCS with multiplexed communication and packet drops.
    \item We provide an optimal switching policy using Deep $\qfactor$ Learning that utilizes the proposed $\varepsilon$-greedy algorithm and ensures MSS, not just after training the NN, but also in the training phase.   
\end{enumerate}

The rest of the paper is organized as follows. Section \ref{sec: MSS problem setup} introduces the problem setup and communication constraints. In Section \ref{sec: MSS switching strategy}, we propose a switching strategy and formulate the problem as a Markovian Jump Linear System (MJLS). We provide necessary conditions for MSS in Section \ref{sec: MSS necc conditions}. Next, we formulate the optimal communication selection problem and provide a solution based on deep reinforcement learning in Section \ref{sec: MSS optimal}. Finally in Section \ref{sec: MSS numerical}, we validate our results on a numerical example.
\section{Problem Setup}\label{sec: MSS problem setup}
\subsection{Plant and Controller Model}
Consider a closed-loop discrete-time linear system 
\begin{equation}\label{eq: plant model}
    \begin{aligned}
        \state_{\Instant+1} &= A \state_{\Instant} + B \control_{\Instant}  \\ 
        \op_{\Instant} &= C \state_{\Instant}
    \end{aligned}
\end{equation}
for all $\Instant \in \Z_{\geq 0}$, where $\state_{\Instant} \in \R[\StDim]$ is the state, $\control_{\Instant} \in \R[\ConDim]$ is the control input and $\op_{\Instant} \in \R[\OpDim]$ is the output at $\Instant^{\tth}$ instant. We make the following assumptions regarding the original closed-loop system.
\begin{assumption}
    The pair $(A, B)$ is controllable and the pair $(A, C)$ is observable.
\end{assumption}
\begin{assumption}
    There exists a state feedback controller of the form
    \begin{equation} \label{eq: closed loop controller}
        \control_{\Instant} = K \state_{\Instant}
    \end{equation} that stabilizes the system \eqref{eq: plant model}.
\end{assumption}

\subsection{Networked System Model}
In this letter, we are interested in an application where the plant and the controller are remotely located. The communication between the plant and the controller occurs over a wireless communication network. The networked system is illustrated in Fig. \ref{fig: NCS schematic}. The networked system dynamics can be written as \begin{equation}\label{eq: networked model}
    \begin{aligned}
        \state_{\Instant+1} &= A \state_{\Instant} + B \hat{\control}_{\Instant} \\ 
        \op_{\Instant} &= C \state_{\Instant}
    \end{aligned}
\end{equation} 
where $\hat{\control}_{\Instant}$ denotes the networked version of the control signal. We proceed by emulation of the controller \eqref{eq: closed loop controller} and use the controller as 
\begin{equation}
    \control_{\Instant} = K \hat{\state}_{\Instant}
\end{equation}
where $\hat{\state}_{\Instant}$ denotes the estimates of the state at the controller end. A more detailed explanation of these quantities is presented later in this section.
\begin{figure}[!htbp]
    \centering
    \begin{tikzpicture}[scale=0.8]
 
\node [draw,
    minimum width=1cm,
    minimum height=1cm, align = left,
]  (controller) at (0,0) {Controller \\ $\control_{\Instant} = K \hat{\state}_{\Instant}$};

\node [draw,
    minimum width=1cm,
    minimum height=1cm, align= center
]  (predictor) at (2,-3) { \small Predictor \\ \small $\hat{\state}_{\Instant} = \begin{cases}
     \state_{\Instant} \\
     A \hat{\state}_{\Instant-1} + B \hat{\control}_{\Instant-1}
\end{cases} $
};
 
\node[cloud, draw,minimum width=1cm, 
    minimum height=1cm] (network) at (3,0) {\small{Network}};

\node [draw,
    minimum width=1cm, 
    minimum height=1cm,
    align = center,
] (system) at (7,0) { 
   \small Plant \\ \small $\state_{\Instant+1} = A \state_{\Instant} + B \hat{\control}_{\Instant} $ \\ 
    $\op_{\Instant} = C \state_{\Instant}  $};
    
\draw[-stealth] (1,0.5) -- (1.7,0.5) node[midway, above]{$\control_{\Instant}$};

\draw[dashed] (1.7,0.5) -- (4,0.5) node[midway, above]{};

\draw[-stealth] (4,0.5) -- (5.1,0.5) node[midway, above]{$\hat{\control}_{\Instant}$};

\draw[-stealth] (5.1,-0.5) -- (4.3,-0.5) node[midway, above]{$\state_{\Instant}$};

\draw[dashed] (4.3,-0.5) -- (3.7,-0.5) node[midway, above]{};

\draw[dashed] (3.7,-0.5) -- (3.7,-1.2) node[midway, above]{};

\draw[-stealth] (3.7,-1.2) -- (3.7,-2.2) node[midway, right]{};


\draw[- stealth] (0,-2.2) -- (0,-0.6) node[midway, right]{$\hat{\state}_{\Instant}$};

    \end{tikzpicture}
    \caption{Schematic of a Networked Control System with information multiplexing and packet drops in the network}
    \label{fig: NCS schematic}
\end{figure}
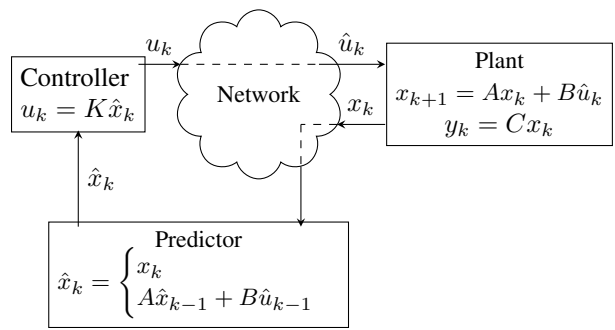
Motivated by real-world communication constraints in the form of bandwidth limitations, we consider a communication constraint as described below. At any time instant, the network scheduler has three choices:
\begin{enumerate}[label = \roman*)]
    \item transmit the control input from the controller to the plant or,
    \item transmit the measured signal from the plant to the predictor or,
    \item not communicate at all.
\end{enumerate}
These three choices are encapsulated in the form of a switch variable $\switch_{\Instant}$
that takes values in a discrete set $\AdmSwitchSet \coloneqq \{-1,0,1\}$
where,
\begin{equation}\label{eq: switch}
\switch_{\Instant} \coloneqq
    \begin{cases}
    1 & \text{if the control is transmitted;}\\
    -1 & \text{if the observation is transmitted;}\\
    0 & \text{if there is no communication.}
    \end{cases}
\end{equation}

We also consider lossy communication in the sense of a packet drop scenario. Consider $\gamma_{k}$ to denote the packet drop event, modeled as independent Bernoulli random variables with probabilities:
\begin{equation}
    \begin{aligned}
        \prob \left( \gamma_{\Instant } = 1 \right) = \delta \quad \text{and} \quad
    \prob \left( \gamma_{\Instant } = 0  \right) = 1 - \delta
    \end{aligned}
\end{equation}
where $\gamma_{\Instant} =1 $ indicates a successful packet transmission and $\gamma_{\Instant} = 0$ indicates failure.

Based on the switching and the packet drop assumptions, the transmitted control information through the network is written as
\begin{equation}\label{eq: transmitted control}
\hat{\control}_{\Instant} \coloneqq
    \begin{cases}
    \control_{\Instant}  & \text{if } \switch_{\Instant}=1 \text{ and } \gamma_{\Instant}=1; \\
    \hat{\control}_{\Instant-1} &  \text{otherwise}.
    \end{cases}
\end{equation}
The coarse estimate of the state at the controlled end is
\begin{equation}
    \hat{\state}_{\Instant} = 
    \begin{cases}
         \state_{\Instant} & \text{if } \switch_{\Instant} = -1 \text{ and } \gamma_{\Instant}=1; \\
        A \hat{\state}_{\Instant-1} + B \hat{\control}_{\Instant-1} & \text{otherwise}
    \end{cases}.
\end{equation}
Define the concatenated state as 
\begin{equation} \label{eq: concatenated state}
    \constate_{\Instant} \coloneqq \begin{pmatrix}
        \state_{\Instant}\\
        \hat{\state}_{\Instant-1} \\
        \hat{\control}_{\Instant-1}
    \end{pmatrix}  
\end{equation}
with $\hat{\state}_{-1} = \state_{0}$ and $\hat{\control}_{\Instant-1} = \bf{0}_{\ConDim}$ respectively.
The overall model of the system is written as
\begin{equation}\label{eq: overall system}
    \constate_{\Instant+1} = 
    \begin{cases}
         \sysmat_{1} \constate_{\Instant}   & \text{if } \switch_{\Instant} = 1 \text{ and } \gamma_{\Instant}=1; \\ 
         \sysmat_{-1} \constate_{\Instant} & \text{if } \switch_{\Instant} = -1 \text{ and } \gamma_{\Instant}=1; \\
         \sysmat_{0} \constate_{\Instant}  & \text{otherwise }
    \end{cases}
\end{equation}
with 
\begin{equation}
    \sysmat_{1} \coloneqq \begin{pmatrix}
        A & BKA & BKB \\
        0 & \idmat{\StDim} & 0\\
        0 &0 & \idmat{\ConDim}
    \end{pmatrix},
\end{equation}
\begin{equation}
    \sysmat_{-1} \coloneqq \begin{pmatrix}
        A & 0 & B \\
        0 & \idmat{\StDim} & 0\\
        0 &0 & \idmat{\ConDim}
    \end{pmatrix}, 
\end{equation}
and 
\begin{equation}
    \sysmat_{0} \coloneqq \begin{pmatrix}
        A & 0 & B \\
        0 & \idmat{\StDim} & 0\\
        0 & 0 & \idmat{\ConDim}
    \end{pmatrix}.
\end{equation}
\section{Switching Strategy}\label{sec: MSS switching strategy}
In this section, we first present the $\varepsilon$ -greedy strategy for switching that ensures MSS. Next, we present the generalized switching probabilities using the $\varepsilon$-greedy algorithm. Lastly, we formulate the system as a Markov Jump Linear System (MJLS), discuss modes of operation in the MJLS, and state our objective. 
\subsection{$\varepsilon$- Greedy Algorithm for Switching}
We employ the $\varepsilon$-greedy switching strategy for the decision-making process \cite{sutton2018reinforcement}. The algorithm consists of two parts:  exploration and exploitation. {\it Exploration} addresses finding new possible solutions, whereas {\it exploitation} addresses utilizing the already known optimal solution. The variable $\varepsilon$ acts as a parameter that weighs the two. Let the per-stage cost be defined by
\begin{equation}\label{eq: per stage cost}
\cost_{\Instant} \coloneqq\state_{\Instant}^{\intercal} Q \state_{\Instant} + \hat{\control}_{\Instant}^{\intercal}R\hat{\control}_{\Instant} + \lambda \switch_{\Instant}^{\intercal} \switch_{\Instant}.
\end{equation}
The switching strategy is defined mathematically as follows:
\begin{equation}\label{eq: switching strategy}
    \switch_{\Instant} = \begin{cases}
        \sim \text{unif}(\{-1,1\}) 
         \hfill \text{if } r_k < \varepsilon \\
         \argmin \limits_{\switch_{t} \in \AdmSwitchSet}  \limsup \limits_{T \to \infty} \frac{1}{T} \expc{  \sum_{t=\Instant}^{T-1} \discfactor^{t}  \cost_{t}}
        & \text{otherwise}
    \end{cases}
\end{equation}
where $r_k \sim \text{unif}\lcrc{0}{1}$ and $\beta \in \loro{0}{1}$ is a discount factor that discounts the cost at future time stages.

When $r_k$ is less than $\varepsilon$ the switching variable $\switch_{\Instant}$ is chosen uniformly randomly from $\{ -1,1\}$. This random selection represents exploration by allowing the system to consider other strategies that might not be optimal for that instance but could provide a better solution over a longer horizon. When $r_{\Instant} \geq \varepsilon$, the switch variable $\switch_{\Instant}$ is determined by the exploitation part, i.e., minimizing the expected cost function. Where the cost function has the following terms:
\begin{enumerate}[label = \roman*)]
\item ${\state_{\Instant}^{\intercal} Q \state_{\Instant}}$, represents a penalty on the state, where $\state_{\Instant}$ is the state at time instant $\Instant$, and $Q \succeq 0$.
\item ${\hat{\control}_{\Instant}^{\intercal} R \hat{\control}_{\Instant}}$ represents the penalty on the control effort, where $\control_{\Instant}$ is the control at time instant $\Instant$, and $R \succ 0$.
\item ${\lambda \switch_{\Instant}^{\intercal} \switch_{\Instant}}$ introduces a penalty for transmission of a packet, with $\lambda$ being a weighing parameter that controls the trade-off between state and control cost versus the transmission cost.
\end{enumerate}
Thus the $\varepsilon$-greedy strategy creates a balance between exploration and exploitation, ensuring that new strategies are explored while utilizing the accumulated information. 
\subsection{Switching Probabilities with $\varepsilon$- Greedy Algorithm}
In this subsection, we discuss in detail the generalized switching probability under the $\varepsilon$- greedy algorithm and analyze the effect of the $\varepsilon$-greedy switching strategy on the MSS of the system. The switching probability distribution function $\switchprob_{g} \in \lcrc{0}{1}^{3}$ is given by
\begin{equation}\label{eq: switch pdf}
    \switchprob_{g} \coloneqq \left( \prob(\switch_{\Instant}=1), \prob(\switch_{\Instant}=0), \prob(\switch_{\Instant}=-1)\right)
\end{equation}
with some switching algorithm $g$.
Given the switching strategy defined in \eqref{eq: switching strategy}, the probabilities of switching states are influenced by the choice of $\varepsilon$. The choice of $\varepsilon$, in turn, decides the balance between exploration and exploitation. With the switching strategy \eqref{eq: switching strategy} the switching probability distribution (denoted as $\switchprob_{\varepsilon}$) can be written as 
\begin{equation}
\begin{aligned}
\switchprob_{\varepsilon} &=\varepsilon \left(\frac{1}{2}, 0 ,\frac{1}{2} \right) + (1-\varepsilon) \left(p_{\Instant}, 1-p_{\Instant}-q_{\Instant} ,q_{\Instant} \right) \\
&=\left( \frac{\varepsilon}{2} + (1-\varepsilon) p_{\Instant}, (1-\varepsilon)(1-p_{\Instant}-q_{\Instant}), \frac{\varepsilon}{2} + (1-\varepsilon) q_{\Instant}  \right). 
\end{aligned}
\end{equation}
Here,
\begin{enumerate}[label=\roman*)]
    \item the first term, $\varepsilon \left(\frac{1}{2}, 0,\frac{1}{2} \right)$, represents the probability distribution in the exploration phase because of the uniform switching between $1$ and $-1$. With probability $\varepsilon$, the switch positions are chosen randomly, with an equal probability ($\frac{1}{2}$) of switching to $1$ or $-1$ and zero probability of remaining in position $0$, and
    \item the second term, $(1-\varepsilon) \left(p, 1-p-q,q \right)$, represents the probability distribution in exploitation phase. With probability $1-\varepsilon$, the switching follows a policy based on the probabilities $p$ and $q$, where:
    \begin{itemize}
        \item $p$ is the probability of switching to position 1,
        \item $q$ is the probability of switching to position $-1$,
         \item $(1-p-q)$ is the probability of remaining in position $0$.
    \end{itemize}
\end{enumerate}
Here, $p,q \in \lcrc{0}{1}$ and $p+q \leq 1$ to ensure that all probabilities sum to 1.
\begin{remark}
    The choice of $\varepsilon$ impacts the MSS of the system which is one of the main interests of this article.
\end{remark}
\subsection{Formulation as a Markovian Jump Linear System}
In this section, we elaborate on the framework of the Markovian Jump Linear System (MJLS) that relates to the system described in \eqref{eq: overall system}. An MJLS is a linear system that goes through random transitions between a finite number of modes, each governed by linear dynamics. The random transitions are governed by Markovian probabilities associated with switching from one mode to another mode. In the problem, the randomness takes place at two levels: the first is at the switching under $\varepsilon$- greedy policy and the second is the random packet drop. With this backdrop, we introduce modes of operation under varying circumstances.

\textbf{System Modes:}
The system can operate in one of several modes at any instance, which is determined by the switch position and the status of packet transmission. These modes represent different scenarios of switching and packet drop. We define the modes of the Markovian switching system as follows:
\begin{equation}\label{eq: mode definition}
    \begin{aligned}
        &\mode 1 : \left( 1, \mathcal{S}, \sysmat_{1} \right),
        \mode 2 : \left(1, \mathcal{F}, \sysmat_{0}  \right)
        \mode 3 : \left(-1, \mathcal{S}, \sysmat_{-1}  \right),\\
        &\mode 4 : \left( -1, \mathcal{F}, \sysmat_{0} \right), \text{ and }
        \mode 5 : \left(0, -, \sysmat_{0} \right)
    \end{aligned}
\end{equation}
where,
\begin{itemize}
    \item the first entry in each tuple indicates the position of the switch, which can be 1,-1, or 0;
    \item the second entry denotes whether the packet transmission was successful ($\mathcal{S}$) or failed ($\mathcal{F}$), and
    \item the third entry corresponds to the system matrix $\sysmat$ that governs the dynamics in that particular mode.
\end{itemize}  
For example, $\mode1$ represents the case where the switch is in position 1, the packet transmission is successful, and the system dynamics are described by the matrix $\sysmat_1$. 
On the other hand, $\mode2$ represents the case where the switch is in position 1 but the packet transmission has failed, thus the system dynamics revert to $\sysmat_0$.

The mode transition probability, $p_{ij}$, represents the likelihood of the system switching from mode $i$ to mode $j$ in the next time step. 
\begin{definition}[Mode Transition Probability] Define the mode transition probability of switching from $i$ to $j$ as,
    $$\prob(\switch_{\Instant+1} = j \vert \switch_{\Instant} = i ) \coloneqq p_{ij}$$
with $i,j \in \{1,2,\hdots, M \}$ where $p_{ij} \in \lcrc{0}{1}$, $\sum_{j=1}^{M} p_{ij} =1$ and $M$ denotes the total number of modes.
\end{definition}
\begin{definition}[Mean Square Stability] \label{def: MSS}
The system \eqref{eq: overall system} is mean square stable if and only if for some $\zeta \geq 1$, $0 < \xi < 1$ and for every $\constate_{0} \in  \R[2\StDim \times \ConDim]$,
    \begin{equation}\label{eq: MSS definition}
        \expc{\constate_{\Instant}^{\intercal} \constate_{\Instant}} \leq \zeta \xi^{\Instant} \constate_{0}^{\intercal} \constate_{0} \quad \text{for all } \Instant \in \Z_{\geq 0}.
    \end{equation}
\end{definition}
{\bf Objective:}
Our goal is to determine the value for $\varepsilon$, given a $\delta \in \lorc{0}{1}$, that ensures the origin of the system \eqref{eq: overall system} remains Mean Square Stable under the switching algorithm \eqref{eq: switching strategy}.
\section{Mean Square Stability of MJLS}\label{sec: MSS necc conditions}
In this section, we propose the methodology used to address the problem, which involves identifying corner cases relevant to switching probabilities. The goal is to determine a value $\Bar{\varepsilon}$, given a fixed $\delta$, which ensures that certain Linear Matrix Inequalities (LMIs) are satisfied for each identified corner case \cite{costa2005discrete}. Note that $\Bar{\epsilon}$ is not a bound but a value that satisfies the LMIs.  Furthermore, we demonstrate that any convex combination of these corner cases also satisfies the MSS conditions derived from the LMIs. The convex combination relates to different switching scenarios in the exploitation phase, implying that the system is MSS irrespective of any switching policy implemented during the exploitation phase. 
\subsection{Corner Cases}
We study this through specific corner cases:
\begin{enumerate}[label =C\arabic*.]
    \item $p=0$ and $q=1$: 
    This case represents the switch being in position $-1$ throughout the entire exploitation phase, indicating that only observations are transmitted. 
    \item $p=1$ and $q=0$:
    This case represents the switch being in position $1$ throughout the entire exploitation phase, indicating that only control signals are transmitted. 
    \item $p=0$ and $q=0$:
    This case depicts the switch remaining in position $0$ for the duration of the exploitation phase, signifying that no transmissions occur. 
\end{enumerate}
The probability of switching in each of the corner cases is tabulated in TABLE \ref{tab: switch prob}. 
\begin{table}[htbp]
    \centering
    \begin{tabular}{c c}
    \toprule
       Case &  Switch probabilities  \\ 
    \midrule
    General & $\left( \frac{\varepsilon}{2} + (1-\varepsilon) p, (1-\varepsilon)(1-p-q), \frac{\varepsilon}{2} + (1-\varepsilon) q  \right)$ \\
         C1 & $\left(\frac{\varepsilon}{2}, 0, 1-\frac{\varepsilon}{2} \right)$  \\
        C2 &  $\left(1-\frac{\varepsilon}{2}, 0,
          \frac{\varepsilon}{2} \right)$  \\
        C3 & $\left(\frac{\varepsilon}{2}, 1- {\varepsilon}, \frac{\varepsilon}{2} \right)$ \\
    \bottomrule
    \end{tabular}
    \caption{Switching probability distribution in different corner cases}
    \label{tab: switch prob}
\end{table}\vspace{-0.5cm}
\subsection{Mode Transition Probabilities for Corner Cases}
Let $\probmat{\ccase} \in \lcrc{0}{1}^{M \times M}$ denote the mode transition probability matrix associated with the corner case $\ccase$, where $\ccase \in \{ 1,2,3\}$. Here, the superscript $\ccase$ is the variable representing each corner case. Each $\probmat{\ccase}_{j} \in \lcrc{0}{1}$ denotes the probability of transitioning from each node to the mode $\mode j$ in the corner case $\ccase$, i.e., $\probmat{\ccase}_{ij} = \probmat{\ccase}_{j} $ for all $i,j \in \{1,2,\hdots,5 \}$. The mode transition probabilities, for each case, are described in Table \ref{tab: mode transition for cases} and depicted in Fig. \ref{fig: packet drop Markov model}. With the given packet transmission success probability $\delta \in \lorc{0}{1}$, we provide the necessary conditions for the MSS of the system using the LMIs given in \cite{costa2005discrete}.
\begin{figure}[!htbp]
    \centering
    \begin{tikzpicture}[
  > = stealth',
  auto,
  prob/.style = {inner sep=1pt, font=\footnotesize},
  node distance=2cm
  ]

  \node[state] (a) {$\mode1$};
  \node[state] (b) [below =of a] {$\mode2$};
  \node[state] (d) [right=of a] {$\mode3$};
  \node[state] (e) [below=of d] {$\mode4$};
 \node[state] (c) [right=of d] {$\mode5$};
 
  \path[->] (b) edge[bend left=15]                  node[prob]{$\probmat{g}_1$} (a)
  edge[loop below] node[prob,red]{$\probmat{g}_2$} (b)
  edge[bend left=15] node[prob,cyan]{$\probmat{g}_3$} (d)
  edge[bend left=15] node[prob,blue]{$\probmat{g}_4$} (e)
  edge[bend left=15] node[prob,brown]{$\probmat{g}_5$} (c);

   \path[->] (a) edge[loop above]                  node[prob]{$\probmat{g}_1$} (a)
  edge[bend left=15] node[prob,red]{$\probmat{g}_2$} (b)
  edge[bend left=15] node[prob,cyan]{$\probmat{g}_3$} (d)
  edge[bend left=15] node[prob,blue]{$\probmat{g}_4$} (e)
  edge[bend left=90] node[prob,brown]{$\probmat{g}_5$} (c);

  \path[->] (d) edge[bend left=15]                  node[prob]{$\probmat{g}_1$} (a)
  edge[bend left=15] node[prob,red]{$\probmat{g}_2$} (b)
  edge[loop above] node[prob,cyan]{$\probmat{g}_3$} (d)
  edge[bend left=15] node[prob,blue]{$\probmat{g}_4$} (e)
  edge[bend left=50] node[prob,brown]{$\probmat{g}_5$} (c);

  \path[->] (e) edge[bend left=15]                  node[prob]{$\probmat{g}_1$} (a)
  edge[bend left=15] node[prob,red]{$\probmat{g}_2$} (b)
  edge[bend left=15] node[prob,cyan]{$\probmat{g}_3$} (d)
  edge[loop below] node[prob,blue]{$\probmat{g}_4$} (e)
  edge[bend left=15] node[prob,brown]{$\probmat{g}_5$} (c);
  
  \path[->] (c) edge[bend right=50]                  node[prob, above]{$\probmat{g}_1$} (a)
  edge[bend left=90] node[prob,red]{$\probmat{g}_2$} (b)
  edge[bend left=15] node[prob,cyan]{$\probmat{g}_3$} (d)
  edge[bend left=15] node[prob,blue]{$\probmat{g}_4$} (e)
  edge[loop below] node[prob,brown]{$\probmat{g}_5$} (c);
\end{tikzpicture}
    \caption{Markov Jump Linear System with associated mode transition probabilities for a general case.}
    \label{fig: packet drop Markov model}
\end{figure}
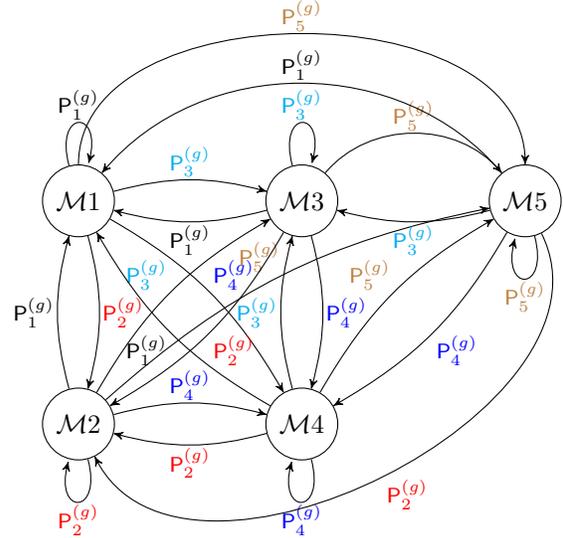

\begin{table*}[htbp]
    \centering
    \begin{tabular}{p{0.05\linewidth} p{0.15\linewidth} p{0.15\linewidth} p{0.15\linewidth} p{0.15\linewidth} p{0.15\linewidth}}
    \toprule 
          & $\probmat{\ccase}_{1} $
          & $\probmat{\ccase}_{2}$ &$\probmat{\ccase}_{3}$ 
          & $\probmat{\ccase}_{4}$ &$\probmat{\ccase}_{5}$  \\ 
    \midrule
         C1 
         & $\delta \frac{\varepsilon}{2}$ 
         & $(1-\delta) \frac{\varepsilon}{2}$ 
         & $\delta  (1 -\frac{\varepsilon}{2})$
         & $(1-\delta)  (1 -\frac{\varepsilon}{2})$
         & $0$  \\
         C2
         &$\delta  (1 -\frac{\varepsilon}{2})$ 
         &$(1- \delta) (1 -\frac{\varepsilon}{2})$
         &$\delta \frac{\varepsilon}{2}$ 
         &$(1-\delta) \frac{\varepsilon}{2}$
         &$0$ \\
        C3
         &$\delta \frac{\varepsilon}{2}$
         &$(1- \delta) (1 -\frac{\varepsilon}{2})$
         &$\delta \frac{\varepsilon}{2}$
         &$(1- \delta) (1 -\frac{\varepsilon}{2})$
         &$1- \varepsilon$ \\
         General         
         &$\delta \left(\frac{\varepsilon}{2} + (1-\varepsilon) p\right)$
         &$(1-\delta)\left(\frac{\varepsilon}{2} + (1-\varepsilon) p\right)$
         & $\delta \left(\frac{\varepsilon}{2} + (1-\varepsilon) q\right)$
         & $(1-\delta) \left(\frac{\varepsilon}{2} + (1-\varepsilon) q\right)$ 
         &$ (1-\varepsilon)(1-p-q)$\\
    \bottomrule
    \end{tabular}
    \caption{Mode transition probabilities for different cases}
    \label{tab: mode transition for cases}
\end{table*}
\begin{theorem}\label{thm: LMI epsilon}
    Given a $\delta \in \lcrc{0}{1}$, with mode transition probability matrix $\probmat{\ccase}$ the origin of the system \eqref{eq: overall system} is MSS if there exist an $\Bar{\varepsilon} \in \loro{0}{1}$ and a symmetric positive definite matrix $V$, such that the following holds
\begin{equation}\label{eq: LMI thm 1}
\begin{aligned}
    \sum_{j=1}^{5} \probmat{\ccase}_{j}  \sysmat_{j}^{\intercal} V \sysmat_{j}  < V 
\end{aligned}
\end{equation}
for all $\ccase \in \{1,2,3\}$.
Furthermore, with this $\varepsilon$, the origin of the system \eqref{eq: overall system} is MSS under the $\varepsilon$-greedy algorithm \eqref{eq: switching strategy}.  
\end{theorem}
\begin{proof}
    First, the MSS conditions given in \cite[Corollary 3.26]{costa2005discrete} are tailored to the specific problem to obtain the LMI \eqref{eq: LMI thm 1}. Suppose there exists an $\Bar{\varepsilon}$ and $V$ such that LMI \eqref{eq: LMI thm 1} hold, then the origin of the system is stable for these corner cases. 
    To prove the origin of the system \eqref{eq: overall system} is MSS under \eqref{eq: switching strategy} with $\varepsilon = \Bar{\varepsilon}$, we prove that the general case can be written as a convex combination of the three corner cases and then prove the LMI holds for the general case. Let $\alpha_{1}, \alpha_{2} \in \lcrc{0}{1}$ and $\alpha_{1}+\alpha_{2}\leq 1$. Taking the convex combination of mode transition probabilities for all corner cases (see TABLE \ref{tab: mode transition for cases}),
\begin{equation}\label{eq: cvx comb}
        \begin{aligned}
           & \alpha_{1} \left(  \delta \frac{\varepsilon}{2} ,
          (1-\delta) \frac{\varepsilon}{2}, 
          \delta  (1 -\frac{\varepsilon}{2}),
          (1-\delta)  (1 -\frac{\varepsilon}{2}),
          0  \right)  \\
          + & \alpha_{2} \left(  \delta  (1 -\frac{\varepsilon}{2}) ,
         (1- \delta) (1 -\frac{\varepsilon}{2}),
         \delta \frac{\varepsilon}{2} ,
         (1-\delta) \frac{\varepsilon}{2},
         0  \right) \\
         + & \left( 1- \alpha_{1} - \alpha_{2} \right) 
         \bigg (  \delta \frac{\varepsilon}{2},
         (1- \delta) (1 -\frac{\varepsilon}{2}),
         \delta \frac{\varepsilon}{2}, \\
         &
         (1- \delta) (1 -\frac{\varepsilon}{2}),
         1- \varepsilon \bigg) \\
          = &
         \bigg( \delta \left(\frac{\varepsilon}{2} + (1-\varepsilon) \alpha_{2}\right),
         (1-\delta)\left(\frac{\varepsilon}{2} + (1-\varepsilon) \alpha_{2}\right),\\
         &
          \delta \left(\frac{\varepsilon}{2} + (1-\varepsilon) \alpha_{1}\right),
          (1-\delta) \left(\frac{\varepsilon}{2} + (1-\varepsilon) \alpha_{1}\right) ,\\
          &
          (1-\varepsilon)(1-\alpha_{2}-\alpha_{1}) 
          \bigg).
        \end{aligned}
    \end{equation}
Comparing \eqref{eq: cvx comb} with the mode transition probability of the general case in TABLE \ref{tab: mode transition for cases}, we have $\alpha_{1} = q$ and $\alpha_{2}=p$.  \\
To prove that \eqref{eq: overall system} is MSS for any $p,q \in  \lcrc{0}{1}$ and $p+q \leq 1$, let $\probmat{g}$ be the general mode transition probability matrix. Then 
\begin{equation*}
    \probmat{g}_{j} = q \probmat{1}_{j} + p \probmat{2}_{j} + (1-q-p)  \probmat{3}_{j}
\end{equation*}
for all $j \in \{1,2,\hdots,5 \}$. 
\begin{equation*}
    \begin{aligned}
        &\sum_{j=1}^{5} \probmat{g}_{j}  \sysmat_{j}^{\intercal} V \sysmat_{j} \\
        =& \sum_{j=1}^{5} \left( q \probmat{1}_{j} + p \probmat{2}_{j} + (1-q-p)  \probmat{3}_{j} \right) \sysmat_{j}^{\intercal} V \sysmat_{j}\\
        =  & q \sum_{j=1}^{5}   \probmat{1}_{j} \sysmat_{j}^{\intercal} V \sysmat_{j} +  p\sum_{j=1}^{5}   \probmat{2}_{j} \sysmat_{j}^{\intercal} V \sysmat_{j} \\
        + & (1-q-p) \sum_{j=1}^{5}   \probmat{3}_{j}  \sysmat_{j}^{\intercal} V \sysmat_{j} \\
        \stackrel{\text{(a)}}{<} & \, q V + p V + (1-q-p) V 
        =  V
    \end{aligned}
\end{equation*}
The inequality (a) is from \eqref{eq: LMI thm 1} and the fact that all quantities on both sides of \eqref{eq: LMI thm 1} are non-negative. 
Hence, if $\Bar{\varepsilon}$ satisfies LMI \eqref{eq: LMI thm 1} for three corner cases then, \eqref{eq: LMI thm 1} is satisfied for any general switching strategy \eqref{eq: switching strategy}.
\end{proof}

To determine the value of $\varepsilon$ that satisfies the LMIs required for MSS, we utilize a method involving Semi-Definite Programming (SDP) solvers and the bisection method.  Based on Theorem \ref{thm: LMI epsilon}, we set up the necessary LMIs involving symmetric positive definite matrix $V$. These LMIs establish the conditions for MSS that the system must satisfy. We employ an SDP solver to numerically solve the formulated LMIs. Given the dependence of $V$ matrix on $\varepsilon$, we apply the bisection method to determine the value of $\varepsilon \, (\Bar{\varepsilon} )$ that satisfies all LMIs. Starting with an initial range for $\varepsilon \in \lcrc{0}{1}$, the bisection method iteratively narrows down to $\Bar{\varepsilon}$ by checking the existence of solutions of LMIs at midpoints within the range.
\section{Optimal Switching Strategy} \label{sec: MSS optimal}
We intend to find an optimal switching policy that satisfies the necessary conditions for MSS established earlier and minimizes the average cost involving penalties on state, control, and communication. That is, to find
\begin{equation}\label{eq: avg cost problem}
\begin{aligned}
    & \min_{\{\switch_{\Instant}\}}  &&  \limsup_{T \to \infty} \frac{1}{T} \expc{  \sum_{\Instant=0}^{T-1}  \cost_{\Instant}},
\end{aligned}
\end{equation}
which can be written as a discounted cost problem 
\begin{equation}\label{eq: discounted cost problem}
\begin{aligned}
    & \min_{\{\switch_{\Instant}\}}  &&  \limsup_{T \to \infty} \frac{1}{T} \expc{  \sum_{\Instant=0}^{T-1} \discfactor^{\Instant}  \cost_{\Instant}}
\end{aligned}
\end{equation}
using \cite[Lemma 5.3.1]{hernandez2012discrete}.
This transformation allows the application of reinforcement learning techniques for solutions. 
\subsection{$\qfactor$-Learning}
The sequential decision-making problem is treated as a Markov Decision Process (MDP). The MDP has three components: i) state, ii) action, and iii) reward. The state space consists of $\{ \state_{\Instant} \} \in \R[\StDim]$ and the action space is $\switch_{\Instant} \in \{-1,0, 1 \}$. Define per-stage reward as
\begin{equation}
    \label{eq: reward per-stage}
    \reward_\Instant \coloneqq -\left( \state_{\Instant}^{\intercal} Q \state_{\Instant} + \hat{\control}_{\Instant}^{\intercal}R\hat{\control}_{\Instant} + \lambda \switch_{\Instant}^{\intercal} \switch_{\Instant} \right).
\end{equation}
The system transitions from state $\state_{\Instant}$ to $\state_{\Instant+1}$ based on action $\switch_{\Instant}$ and receives reward $\reward_{\Instant+1}$. To evaluate the performance of a given policy, we use the action value function, denoted as $\qfactor_\policy(\state,\switch)$. For a policy $\policy$, which maps states to probabilities of selecting actions, the action value function is defined as:
\begin{equation}
    \label{eq: q-factor}
    \qfactor_{\policy}(\state, \switch) = \expc{ \sum_{\Instant=0}^{T-1} \discfactor^{\Instant} \reward_{\Instant+1} | \state, \switch}.
\end{equation}
The optimal $\qfactor$ value (denoted as $\qfactor_{*}(\state_{\Instant}, \switch_{\Instant})$) satisfies Bellman's principle of optimality:
\begin{equation}\label{eq: Bellman}
    \qfactor_{*}(\state_{\Instant}, \switch_{\Instant}) = \expc{\reward_{\Instant} + \discfactor \qfactor_{*}(\state_{\Instant+1}, \switch_{\Instant+1})}.
\end{equation}
The optimal policy can be found iteratively in the case of discrete sets of states and actions by the value iteration method \cite{sutton2018reinforcement}. In this problem, the state space is continuous, and this method cannot be directly applied. Hence, we use an advanced deep reinforcement learning technique using a neural network.
\subsection{Deep $\qfactor$-Learning}
We approach the solution of the problem \eqref{eq: discounted cost problem} using deep reinforcement learning \cite{mnih2015human}. A Deep Q Network (DQN) estimates the $\qfactor$-values for each state and action. Thereby approximates the function $\qfactor(\state,\switch; \theta)$, where the function is parameterized by $\theta$. The state acts as the input to the DQN, and the output corresponding to each action is the approximated $\qfactor$-value. The loss is derived using Bellman's principle as
\begin{equation}
 \label{eq: DQN loss}
 \text{loss} \coloneqq \left[ \qfactor(\state_{\Instant}, \switch_{\Instant}; \theta) - \left({\reward_{\Instant} + \discfactor \qfactor(\state_{\Instant+1}, \switch_{\Instant+1}; \theta)} \right) \right]^{2}.
\end{equation}
Here, two forward passes are made through the network to get the $\qfactor$-values for the current and next states. To avoid this, Mnih et al. introduced the replay memory framework and another target network replicating the actual network but updated only after a few iterations \cite{mnih2015human}. Algorithm 1 presents the deep reinforcement learning technique used.
\begin{algorithm}
  \caption{Deep reinforcement learning}
  \begin{algorithmic}[1]
  \STATE {Initialize the replay memory $D$ to capacity $K$}
  \STATE Initialize the policy network $Q$ with random weights $\theta_{0}$
  \STATE Copy the policy network as the target network $\hat{Q}$ with weights $\theta^{-} = \theta_{0}$
  \STATE Initialize $\state_{0}$
  \FOR{$\Instant= 0,1, \hdots, N$}
  \STATE Select an action ($\switch_{\Instant}$ via switching algorithm \eqref{eq: switching strategy} with the $\varepsilon$ based on Theorem \ref{thm: LMI epsilon})
  \STATE Store the experience $e_{\Instant} \coloneqq (\switch_{\Instant}, \state_{\Instant}, \reward_{\Instant},  \state_{\Instant+1})$ in the replay memory 
  \STATE Sample random batch from the replay memory
  \STATE Calculate the loss between output $\qfactor$-values and target $\qfactor$-values as \eqref{eq: DQN loss}
  \STATE Update the policy network using the backpropagation
  \STATE Update the target network after a predefined number of steps
  \ENDFOR
  \end{algorithmic}
\end{algorithm}
\section{Numerical Experiments}\label{sec: MSS numerical}
In this section, we validate our approach on system \eqref{eq: overall system} with the following data:
\begin{equation*}
    A = \begin{pmatrix}  1 &0.1 \\ 0 & 1 \end{pmatrix}, \, B = \begin{pmatrix}
        0 \\ 1
    \end{pmatrix}, \text{ and } K = \begin{pmatrix}
        -0.012 \\ -0.07
    \end{pmatrix}^{\intercal}
\end{equation*}
where the controller gain $K$ stabilizes as per Assumption 2.

First, we fix the packet transmission success probability $\delta \in \lorc{0}{1}$, which is typically determined by the communication system. For a given $\delta$, we determine the corresponding value of $\varepsilon$ ($\Bar{\varepsilon}$) that ensures MSS for all corner cases. The relationship between the packet transmission success probability $\delta$ and the $\bar{\varepsilon}$ required for ensuring MSS is depicted in Fig. \ref{fig: eps-delta}.
\begin{figure}[!htbp]
    \centering
\includegraphics[width=0.48\textwidth]{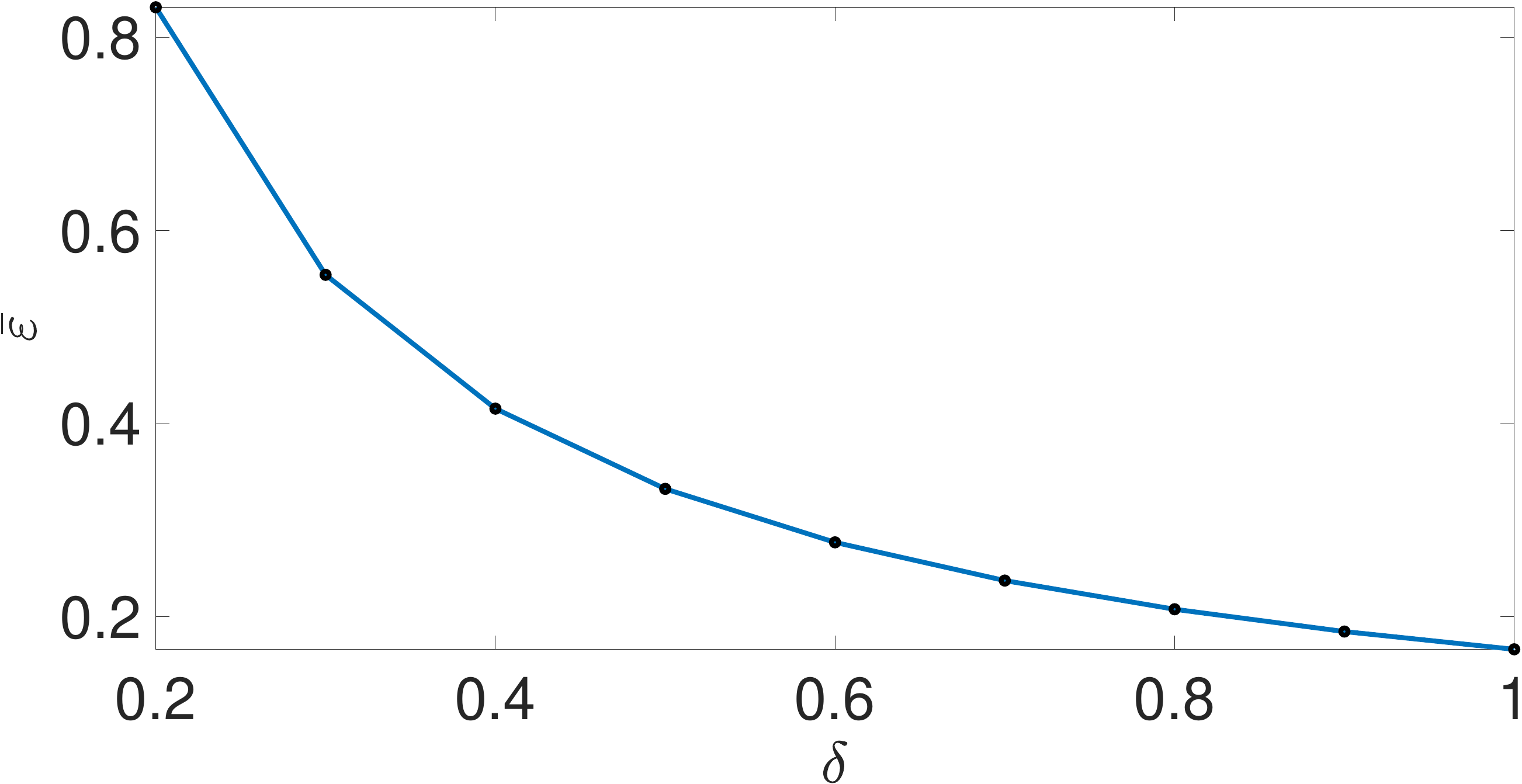}
    \caption{ $\Bar{\varepsilon}$ that satisfies MSS conditions for all three corner cases for different values of $\delta$ }
    \label{fig: eps-delta}
\end{figure}
Next, for a fixed packet transmission success probability $\delta$, we determine the corresponding value of $\varepsilon$ ($\Bar{\varepsilon}$) that ensures the MSS of the system. We then illustrate the associated corner cases and the convex region formed by these corner cases. 
We implement the Deep Q-Network according to Algorithm 1 with the following specifications. The initial state conditions are randomly sampled from the uniform distribution within the $[-10,10]^{2}$ range. A replay memory size of $1000$ is utilized to store experiences. The policy network architecture consists of two hidden layers, with $1024$ and $256$ neurons, respectively. Since there are three possible actions, the output layer consists of three neurons. Mini-batches of size $32$ are used for training. We fix $\delta = 0.8$ and the exploration parameter is set to $ \Bar{\varepsilon} =0.2$. The learning rate is fixed at $0.001$. The network undergoes training for $800$ episodes. The moving average of reward in the training phase shows the DQN's learning over several episodes and then the learning settles (see Fig. \ref{fig: training reward}).
\begin{figure}[!htbp]
    \centering
    \includegraphics[width=0.5\textwidth]{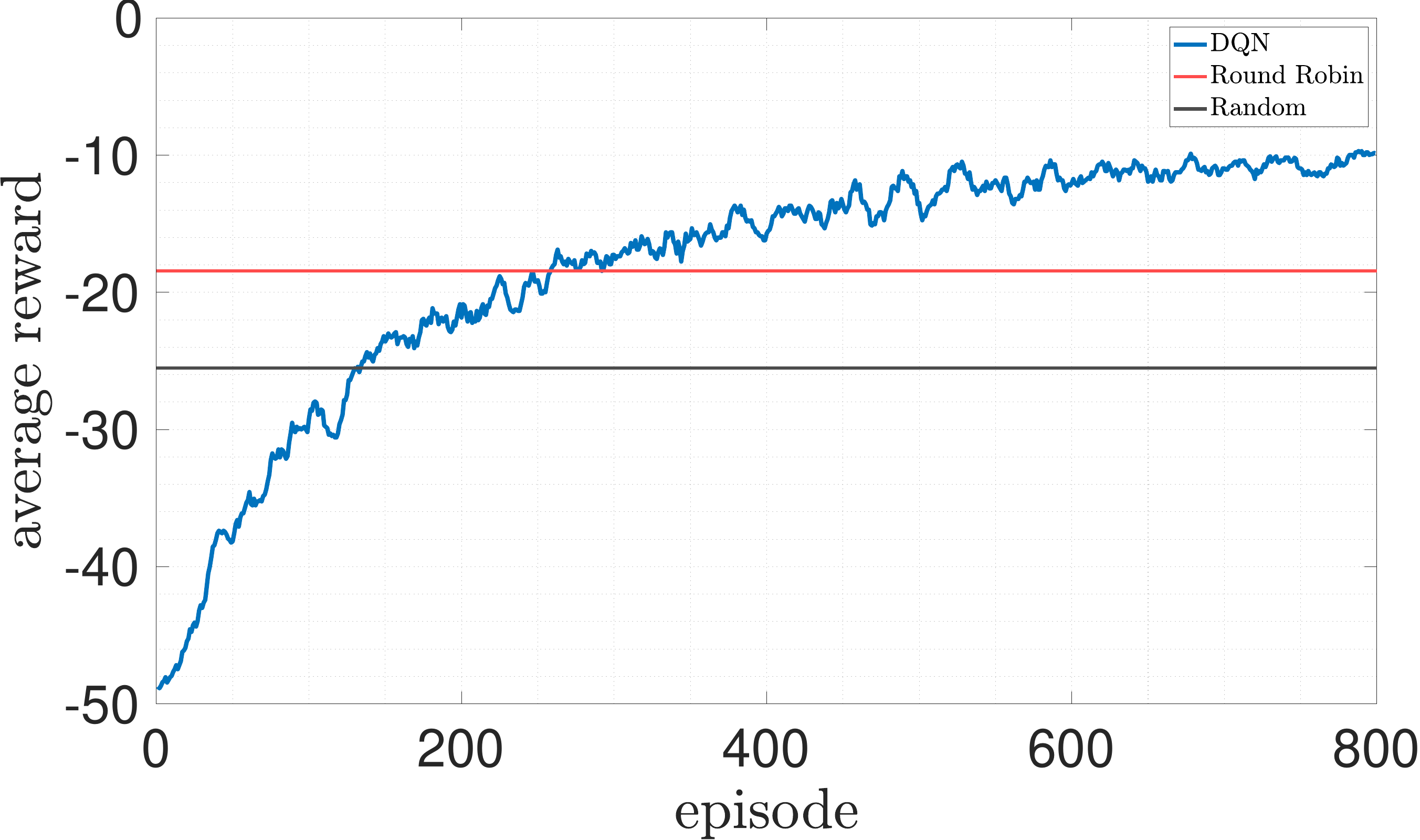}
    \caption{Comparison of average reward ($\lambda=0.5$) of DQN method with round robin and random switching scheme}
    \label{fig: training reward}
\end{figure}
To assess the efficacy of our approach, we compare our results with two other methods: round-robin and random selection. Each method undergoes testing under identical initial conditions and parameters. The DQN method performs better than the other two methods. The round-robin method yields an average reward of $-18.51$, while the random method results in an average reward of $-25.6$. These comparisons underscore the effectiveness of the DQN after about $250$ episodes of training.
In this subsection, we analyze the MSS property with the proposed $\varepsilon$- greedy algorithm, under different packet transmission success probability and corresponding $\varepsilon$ found using Theorem \ref{thm: LMI epsilon}. Figure \ref{fig: low delta high eps} represents the case with low transmission success probability that requires a high value of $\varepsilon$ to be MSS. Conversely, with low transmission success probability and low value of $\varepsilon$, we can observe in Figure \ref{fig: low delta low eps} that the system is not MSS. As depicted in Figure \ref{fig: high delta low eps}, with high transmission success probability, a low value of $\varepsilon$ is sufficient for achieving MSS.
\begin{figure}[htbp]
\centering
\includegraphics[width=0.48\textwidth]{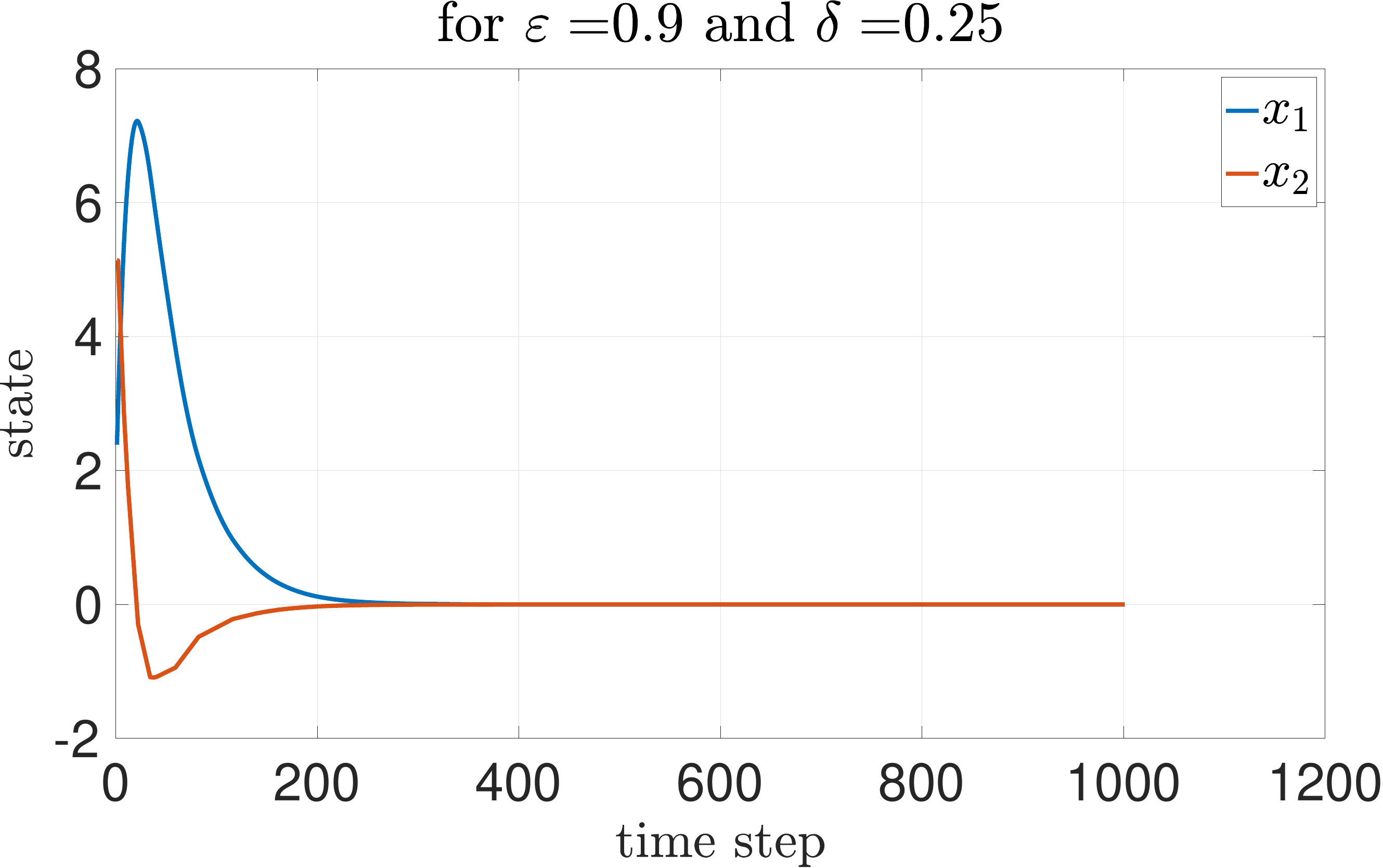}
\caption{Low success probability, with high $\varepsilon$}
\label{fig: low delta high eps}
\end{figure}
\begin{figure}[htbp]
\centering
\includegraphics[width=0.48\textwidth]{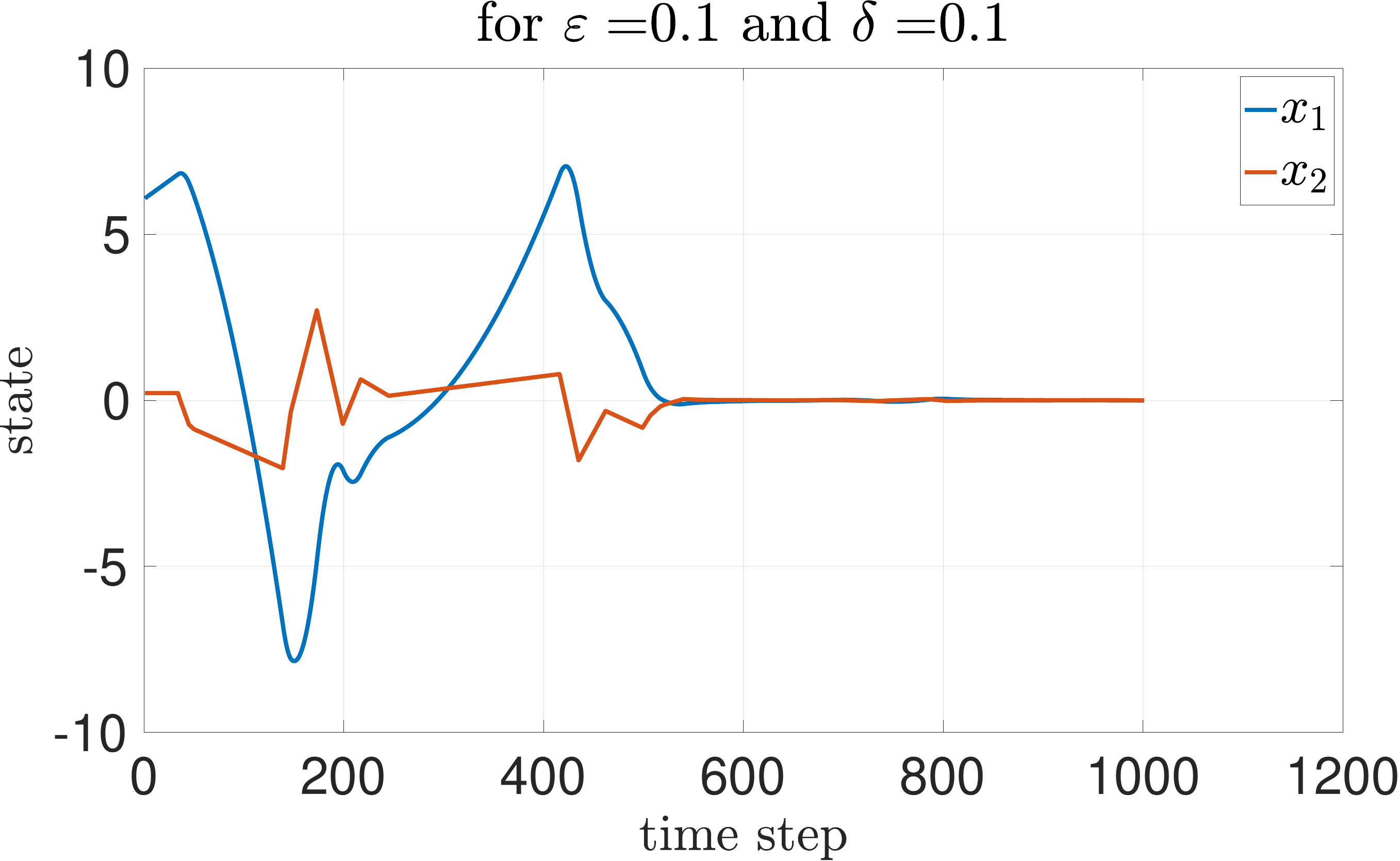}
\caption{Low success probability, with low $\varepsilon$}
\label{fig: low delta low eps}
\end{figure}
\begin{figure}[htbp]
\centering
\includegraphics[width=0.48\textwidth]{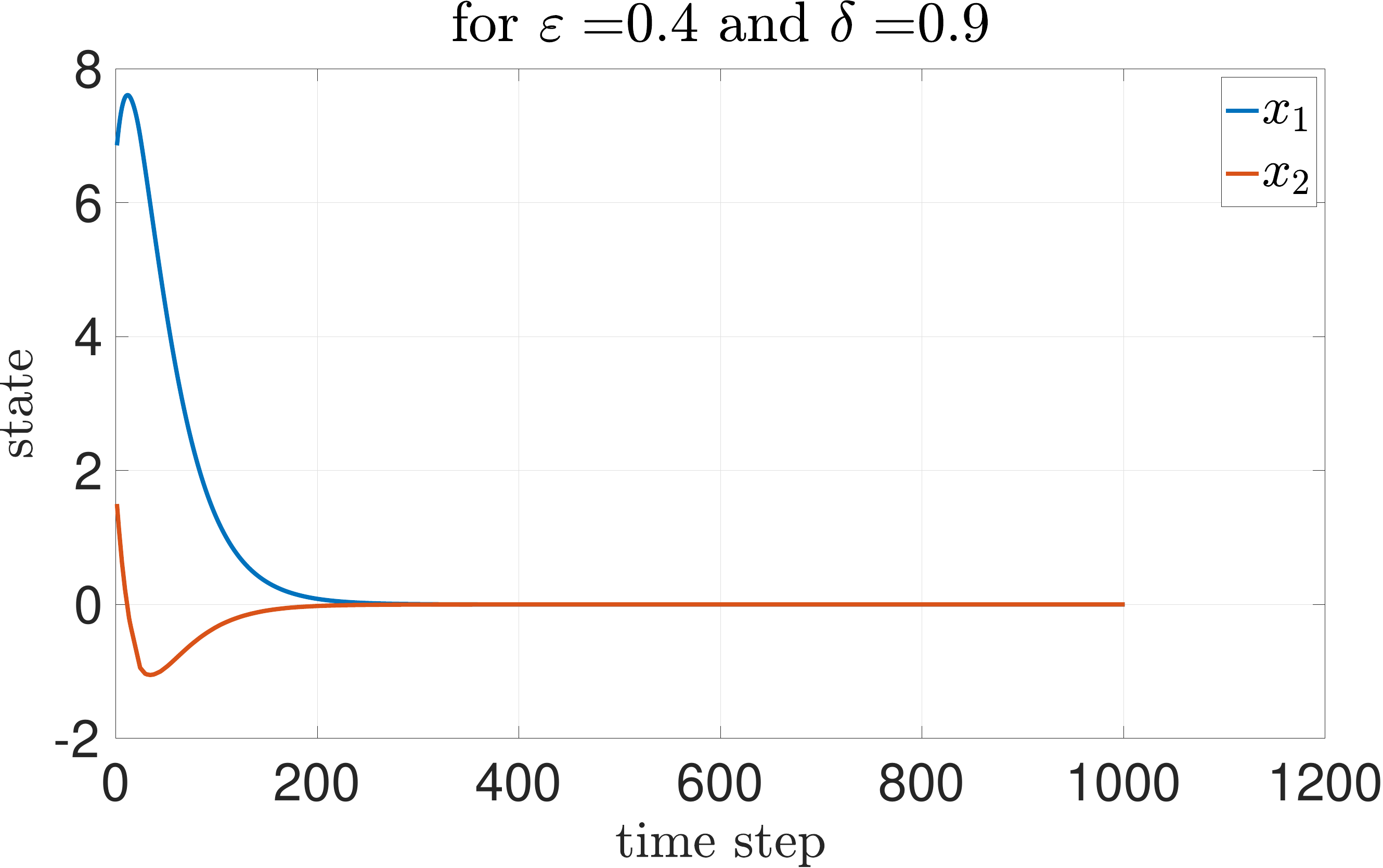}
\caption{High success probability, with low $\varepsilon$}
\label{fig: high delta low eps}
\end{figure}
\section{Conclusion}
In this letter, we have proposed a modified $\varepsilon$-greedy algorithm for selecting communication direction in a multiplexed NCS. We have established the necessary conditions for the mean square stability (MSS) of an NCS with multiplexed communication and packet drops. We provided an optimal switching policy using Deep Q-Learning that utilizes the proposed $\varepsilon$-greedy algorithm and ensures MSS for training the neural network.

\addtolength{\textheight}{-12cm}  



%


\bibliographystyle{ieeetr}
\bibliography{ref}

\begin{thebibliography}{10}

\bibitem{bemporad2010networked}
A.~Bemporad, M.~Heemels, M.~Johansson, {\em et~al.}, {\em Networked control systems}, vol.~406.
\newblock Springer, 2010.

\bibitem{ge2017distributed}
X.~Ge, F.~Yang, and Q.-L. Han, ``Distributed networked control systems: A brief overview,'' {\em Information Sciences}, vol.~380, pp.~117--131, 2017.

\bibitem{elia2004bode}
N.~Elia, ``When bode meets shannon: Control-oriented feedback communication schemes,'' {\em IEEE transactions on Automatic Control}, vol.~49, no.~9, pp.~1477--1488, 2004.

\bibitem{heemels2010networked}
W.~M.~H. Heemels, A.~R. Teel, N.~Van~de Wouw, and D.~Ne{\v{s}}i{\'c}, ``Networked control systems with communication constraints: Tradeoffs between transmission intervals, delays and performance,'' {\em IEEE Transactions on Automatic control}, vol.~55, no.~8, pp.~1781--1796, 2010.

\bibitem{dolk2017event}
V.~Dolk and M.~Heemels, ``Event-triggered control systems under packet losses,'' {\em Automatica}, vol.~80, pp.~143--155, 2017.

\bibitem{sandberg2022secure}
H.~Sandberg, V.~Gupta, and K.~H. Johansson, ``Secure networked control systems,'' {\em Annual Review of Control, Robotics, and Autonomous Systems}, vol.~5, pp.~445--464, 2022.

\bibitem{athans1977uncertainty}
M.~Athans, R.~Ku, and S.~Gershwin, ``The uncertainty threshold principle: Some fundamental limitations of optimal decision making under dynamic uncertainty,'' {\em IEEE Transactions on Automatic Control}, vol.~22, no.~3, pp.~491--495, 1977.

\bibitem{imer2004optimal}
O.~C. Imer, S.~Yuksel, and T.~Ba{\c{s}}ar, ``Optimal control of dynamical systems over unreliable communication links,'' {\em IFAC Proceedings Volumes}, vol.~37, no.~13, pp.~991--996, 2004.

\bibitem{schenato2007foundations}
L.~Schenato, B.~Sinopoli, M.~Franceschetti, K.~Poolla, and S.~S. Sastry, ``Foundations of control and estimation over lossy networks,'' {\em Proceedings of the IEEE}, vol.~95, no.~1, pp.~163--187, 2007.

\bibitem{liu2004kalman}
X.~Liu and A.~Goldsmith, ``Kalman filtering with partial observation losses,'' in {\em 2004 43rd IEEE Conference on Decision and Control (CDC)(IEEE Cat. No. 04CH37601)}, vol.~4, pp.~4180--4186, IEEE, 2004.

\bibitem{sinopoli2004kalman}
B.~Sinopoli, L.~Schenato, M.~Franceschetti, K.~Poolla, M.~I. Jordan, and S.~S. Sastry, ``Kalman filtering with intermittent observations,'' {\em IEEE transactions on Automatic Control}, vol.~49, no.~9, pp.~1453--1464, 2004.

\bibitem{you2011mean}
K.~You, M.~Fu, and L.~Xie, ``Mean square stability for kalman filtering with markovian packet losses,'' {\em Automatica}, vol.~47, no.~12, pp.~2647--2657, 2011.

\bibitem{mesquita2012redundant}
A.~R. Mesquita, J.~P. Hespanha, and G.~N. Nair, ``Redundant data transmission in control/estimation over lossy networks,'' {\em Automatica}, vol.~48, no.~8, pp.~1612--1620, 2012.

\bibitem{varma2023transmission}
V.~S. Varma, R.~Postoyan, D.~E. Quevedo, and I.-C. Morărescu, ``Event-triggered transmission policies for nonlinear control systems over erasure channels,'' {\em IEEE Control Systems Letters}, vol.~7, pp.~2113--2118, 2023.

\bibitem{maity2021optimal}
D.~Maity, M.~H. Mamduhi, S.~Hirche, and K.~H. Johansson, ``Optimal lqg control of networked systems under traffic-correlated delay and dropout,'' {\em IEEE Control Systems Letters}, vol.~6, pp.~1280--1285, 2021.

\bibitem{zhang2012network}
L.~Zhang, H.~Gao, and O.~Kaynak, ``Network-induced constraints in networked control systems—a survey,'' {\em IEEE transactions on industrial informatics}, vol.~9, no.~1, pp.~403--416, 2012.

\bibitem{molin2009lqg}
A.~Molin and S.~Hirche, ``On lqg joint optimal scheduling and control under communication constraints,'' in {\em Proceedings of the 48h IEEE Conference on Decision and Control (CDC) held jointly with 2009 28th Chinese Control Conference}, pp.~5832--5838, IEEE, 2009.

\bibitem{leong2020deep}
A.~S. Leong, A.~Ramaswamy, D.~E. Quevedo, H.~Karl, and L.~Shi, ``Deep reinforcement learning for wireless sensor scheduling in cyber--physical systems,'' {\em Automatica}, vol.~113, p.~108759, 2020.

\bibitem{sutton2018reinforcement}
R.~S. Sutton and A.~G. Barto, {\em Reinforcement learning: An introduction}.
\newblock MIT press, 2018.

\bibitem{costa2005discrete}
O.~L.~V. Costa, M.~D. Fragoso, and R.~P. Marques, {\em Discrete-time Markov jump linear systems}.
\newblock Springer Science \& Business Media, 2005.

\bibitem{hernandez2012discrete}
O.~Hern{\'a}ndez-Lerma and J.~B. Lasserre, {\em Discrete-time Markov control processes: basic optimality criteria}, vol.~30.
\newblock Springer Science \& Business Media, 2012.

\bibitem{mnih2015human}
V.~Mnih, K.~Kavukcuoglu, D.~Silver, A.~A. Rusu, J.~Veness, M.~G. Bellemare, A.~Graves, M.~Riedmiller, A.~K. Fidjeland, G.~Ostrovski, {\em et~al.}, ``Human-level control through deep reinforcement learning,'' {\em nature}, vol.~518, no.~7540, pp.~529--533, 2015.

\end{thebibliography}

\end{document}